\documentclass[10pt, a4paper]{article}

\usepackage{geometry}
\usepackage{CJKutf8}
\usepackage{authblk}
\usepackage{listings}

\usepackage{amssymb}
\usepackage{amsmath}
\usepackage{amsthm}
\usepackage{amsfonts}
\usepackage{mathtools}
\usepackage{mathrsfs}
\usepackage{bm} 

\usepackage{graphicx,subfig}
\usepackage{dcolumn}
\usepackage[ruled,lined,algonl]{algorithm2e}
\usepackage{color}
\usepackage{autobreak}
\allowdisplaybreaks

\theoremstyle{plain}
\newtheorem{theorem}{Theorem}

\newtheorem{lemma}{Lemma}
\newtheorem{proposition}{Proposition}

\theoremstyle{definition}

\theoremstyle{remark}
\newtheorem{remark}{Remark}

\DeclareMathOperator{\Det}{Det}
\DeclareMathOperator{\Tr}{Tr}
\DeclareMathOperator{\res}{res}

\newcommand{\pset}[1]{\mathcal{#1}}

\begin{document}
\begin{CJK}{UTF8}{gbsn}

\title{Equilibria and their stability in an asymmetric duopoly model of Kopel}


\author[a]{Xiaoliang Li}

\author[a]{Kongyan Chen\thanks{Corresponding author: chenkongyan@dgcu.edu.cn}}

\affil[a]{School of Digital Economics, Dongguan City University, Dongguan 523419, China}


\date{}
\maketitle

\begin{abstract}
In this paper, we investigate the equilibria and their stability in an asymmetric duopoly model of Kopel by using several tools based on symbolic computations. We explore the possible positions of the equilibria in Kopel's model. We discuss the possibility of the existence of multiple positive equilibria and establish a necessary and sufficient condition for a given number of equilibria to exist. Furthermore, if the two duopolists adopt the best response reactions or homogeneous adaptive expectations, we establish rigorous conditions for the existence of distinct numbers of positive equilibria for the first time.

\vspace{10pt}
\noindent\emph{Keywords: duopoly; Kopel's model; equilibrium; local stability; symbolic computation}
\end{abstract}

\section{Introduction}

Duopoly is an intermediate market between monopoly and perfect competition \cite{guirao_extensions_2010}.
The study of duopolistic competitions lies at the core of the field of industrial organization.  Since about three decades ago, economists have been making efforts to model duopolistic competitions by discrete dynamical systems of form with unimodal reaction functions. Among them, Kopel \cite{kopel_simple_1996} proposed a famous duopoly model by assuming a linear demand function and nonlinear cost functions, which can be described as the following map.
\begin{equation}\label{eq:map-kopel}
\left\{
\begin{split}
&x(t+1)=(1-a)x(t)+ auy(t)(1-y(t)),\\
&y(t+1)=(1-b)y(t)+ bvx(t)(1-x(t)),
\end{split}\right.
\end{equation}
where $0< a,b\leq 1$ and $u,v>0$. In this map, $x(t)$ and $y(t)$ denote the output quantities of the two duopolists, respectively. Moreover, $a$ and $b$ represent the adjustment coefficients of the adaptive expectations of the two firms, while $u$ and $v$ measure the intensity of the positive externality the actions of the player exert on the payoff of its rival. 

Kopel's model is mathematically interesting because it couples two standard logistic maps together. Afterward, a lot of contributions were made to intensive investigations, extensions, and generalizations of the model. For example, Agiza \cite{agiza_analysis_1999} established rigorous conditions for the stability of the equilibria and investigated the bifurcations and chaos of Kopel's model in the special case of $a=b$ and $u=v$. Bischi et al.\ \cite{bischi_multistability_2000} proved the existence of the multistability and cycle attractors in map \eqref{eq:map-kopel}. Bischi and Kopel \cite{bischi_equilibrium_2001} used the method of critical curves to analyze the global bifurcations and illustrate the qualitative changes in the topological structure of the basins of Kopel's model. Chaotic dynamics in map \eqref{eq:map-kopel} seems to be observed through numerical simulations by many researchers. But, this is not rigorous proof. In this regard, Wu et al.\ \cite{wu_complex_2010} gave the rigorous computer-assisted verification for the existence of the chaotic dynamics in Kopel's model by using the topological horseshoe theory. Moreover, C{\'a}novas and Mu{\~n}oz-Guillermo \cite{Canovas2018O} analytically proved the existence of chaos if the firms in the model of Kopel are homogeneous. Furthermore, Zhang and Gao \cite{zhang_equilibrium_2019} extended the model by assuming that each firm could forecast its rival’s output through a straightforward extrapolative foresight technology. Elsadany and Awad \cite{elsadany_dynamical_2016} modified Kopel's game by assuming one player uses a naive expectation whereas the other employs an adaptive technique and applied the feedback control method to control the chaotic behavior. Torcicollo \cite{torcicollo_dynamics_2013} generalized map \eqref{eq:map-kopel} by introducing the self-diffusion and cross-diffusion terms. 

In the rich literature regarding the model of Kopel, there are nearly no analytical discussions on the equilibria and their stability in the asymmetric case of $u\neq v$, which surprises us much. The primary goal of our study is to fill this gap. The major obstacle to the analytical study on the asymmetric case of $u\neq v$ is that the closed-form equilibria can not be obtained explicitly. We employ several tools based on symbolic computations such as the triangular decomposition \cite{wang_elimination_2001} and resultant \cite{Mishra1993A} methods to get around this obstacle. We explore the possibility of the existence of multiple positive equilibria, which has received a lot of attention from economists. We establish a necessary and sufficient condition for a given number of equilibria to exist. If the two duopolists adopt the best response reactions ($a=b=1$) or homogeneous adaptive expectations ($a=b$), we acquire rigorous conditions for the existence of distinct numbers of positive equilibria for the first time. 

%
%
%

The rest of this paper is structured as follows. In Section 2, we explore the number of the equilibria in map \eqref{eq:map-kopel} and their possible positions. In Section 3, the local stability of the equilibria is analytically investigated. Concluding remarks are provided in Section 4.

\section{Equilibrium Analysis}


To the best of our knowledge, Li et al.\ \cite{li_complex_2022} first explored the number of equilibria in the asymmetric case of $u\neq v$. In this section, we aim at providing a more systematic analysis of this case. By setting $x(t+1)=x(t)=x^*$ and $y(t+1)=y(t)=y^*$ in map \eqref{eq:map-kopel}, we obtain the equilibrium equations
\begin{equation}\label{eq:eq-equi}
\left\{
	\begin{split}
	& x^*=uy^*(1-y^*),\\
	& y^*=vx^*(1-x^*).	
	\end{split}
\right.
\end{equation}


Provided that $u,v\leq 4$, it is evident that $[0,1]\times [0,1]$ is a trapping set of map \eqref{eq:map-kopel}, which means that all the trajectories will stay in $[0,1]\times [0,1]$ if the initial belief $(x(0),y(x))$ is taken from $[0,1]\times [0,1]$. In addition, from the following Proposition \ref{prop:all-equi}, one knows that if $uv\geq 1$ (even if $u,v>4$), then the equilibria that the trajectories approach should lie in $[0,1]\times [0,1]$ provided that the initial beliefs are taken from $[0,1]\times [0,1]$. In the proof of Proposition \ref{prop:all-equi}, the triangular decomposition method plays an ambitious role, which extends the Gaussian elimination method to solving polynomial equations. Readers may refer to \cite{wang_elimination_2001,Li2010D} for more information regarding the triangular decomposition. Moreover, the concept of the resultant is needed. The following lemma reveals the main property of the resultant, which can also be found in \cite{Mishra1993A}.

\begin{lemma}\label{lem:res-com}
Let $A$ and $B$ be two univariate polynomials in $x$.
There exist two polynomials $F$ and $G$ in $x$ such that 
$$FA+GB=\res(A,B,x).$$
Furthermore, $A$ and $B$ have common zeros in the field of complex numbers if and only if $\res(A, B, x)=0$.
\end{lemma}

\begin{proposition}\label{prop:all-equi}
	In map \eqref{eq:map-kopel}, all the equilibria remain inside $[0,1]\times [0,1]$ if $uv\geq 1$.
\end{proposition}

\begin{proof}
By using the triangular decomposition method, we can decompose the solutions of Eq.\ \eqref{eq:eq-equi} into zeros of two triangular sets, i.e., $[y^*,x^*]$ and $[T_2,T_1]$,
where
\begin{align*}
	&T_1=u v^2 x^{*3}- 2\, u v^{2} x^{*2}+\left(u v^{2}+u v\right) x^*-u v+1,~~T_2=y^*+vx^{*2} -vx^*.
\end{align*}
The zero of $[y^*,x^*]$ corresponds to the equilibrium $E_0(0,0)$, which obviously lies in $[0,1]\times [0,1]$. Therefore, we focus on the other branch $[T_2, T_1]$. 

Because of the symmetry of $x^*,y^*$ and $u,v$ in Eq.\ \eqref{eq:eq-equi}, we just need to prove any zero of $T_1$ satisfies $x^*\in [0,1]$ if $uv\geq 1$. One can compute that $\res(T_1,1-x^*,x^*)=1$. According to Lemma \ref{lem:res-com}, we know any zero $x^*$ of $T_1$ can not touch the line $x^*=1$ as $u$ and $v$ vary. Therefore, any zero $x^*$ of $T_1$ is smaller than 1. Furthermore, we have $\res(T_1,x^*,x^*)=uv-1$. This means that the sign of a zero of $T_1$ may change as the parameter point $(u,v)$ goes across the curve $uv=1$. It is easy to check that any zero of $T_1$ is positive if $uv>1$ and vice versa. When $uv=1$, we can plug $v=1/u$ into $T_1$ and solve three complex zeros $0$, $1+\sqrt{-u}$, and $1-\sqrt{-u}$. The only real zero $x^*=0$ is in $[0,1]$, obviously. In conclusion, all zeros of $T_1$ remain inside $[0,1]$ if $uv\geq 1$, which completes the proof.
\end{proof}

From an economic point of view, positive equilibria that satisfy $x^*,y^*>0$ are more important and mainly considered in what follows.

\begin{theorem}\label{thm:equi-num}
Denote $R_1=u^2 v^2-4\,u^2 v-4\,u v^2 +18\,uv -27$. Provided that $R_1\neq 0$, three distinct positive equilibria exist if and only if $R_1>0$, and one unique positive equilibrium exists if and only if $R_1<0$ and $uv>1$. Provided that $R_1=0$, two distinct positive equilibria exist if and only if $(u,v)\neq (3,3)$, and one unique positive equilibrium $\left(\frac{2}{3},\frac{2}{3}\right)$ exists if $(u,v)= (3,3)$.
\end{theorem}

\begin{proof}
One can see that $T_2$ is linear with respect to $y^*$, which means that the number of equilibria corresponding to $[T_2, T_1]$ equals that of distinct real zeros of $T_1$. Consequently, we focus on analyzing the real zeros of $T_1$ in the sequel. 

For a univariate polynomial $F(x)$, the multiplicity of a zero $\bar x$ is greater than 1 if $F(\bar x)=0$ and $F'(\bar x)=0$ are simultaneously satisfied, where $F'$ is the derivative of $F$. We have $\res(T_1,T'_1,x^*)=-u^{3} v^{6} R_1$.
Therefore, the nature (e.g., the multiplicity and number) of real zeros of $T_1$ should not change if the parameter point $(u,v)$ is not across the curve $R_1=0$ as the values of $u,v$ vary. As shown in Figure \ref{fig:num-equilibria}, the curve $R_1=0$ (the red curve) divides the parameter set of concern $\{(u,v)\,|\,u,v>0\}$ into two regions. In the dark-gray region defined by $R_1>0$, the nature of the zeros can be determined by checking at some sample point, say $(4,4)$ for example. At $(4,4)$, $T_1=64\, x^{*3}-128\, x^{*2}+80\, x^*-15$, which can be solved by $\frac{3}{4},\frac{5}{8}-\frac{\sqrt{5}}{8},\frac{5}{8}+\frac{\sqrt{5}}{8}$. Hence, three distinct real zeros of $T_1$ exist if $R_1>0$. This implies that three distinct positive equilibria exist in map \eqref{eq:map-kopel} if $R_1>0$.

In the other region defined by $R_1<0$, we similarly select a sample point, e.g., $(2,2)$. We have $T_1=8\, x^{*3}-16\, x^{*2}+12\, x^*-3$, which can be solved by $\frac{1}{2},\frac{3}{4}-\frac{\mathrm{I} \sqrt{3}}{4},\frac{3}{4}+\frac{\mathrm{I} \sqrt{3}}{4}$. Thus, only one real solution of $T_1$ exists if $R_1<0$. However, in the proof of Proposition \ref{prop:all-equi}, it has been proved that the zero of $T_1$ is positive if $uv>1$. Therefore, one unique positive equilibrium exists if $R_1<0$ and $uv>1$.

Furthermore, on the curve $R_1=0$, we have $T_1=T'_1=0$ according to Lemma \ref{lem:res-com}. Then, the multiplicity of some zeros of $T_1$ should be greater than 1, which means that some zeros of $T_1$ become the same as the parameter point $(u,v)$ varies from the region $R_1>0$ to its edge $R_1=0$. The degree of $T_1$ with respect to $x^*$ is 3, thus the multiplicity of a zero of $T_1$ can be taken as 3 at maximum. In this case, $T_1=0$, $T_1'=0$, and $T''_1=0$ should be fulfilled simultaneously. One can compute $\res(T_1,T_1'',x^*)=-8\, u^{3} v^{6} \left(2\, u v^{2}-9\, u v+27\right)$. 
By solving $2\, u v^{2}-9\, u v+27=0$ and $R_1=0$, we obtain $(u,v)=(3,3)$, where $T_1$ has one unique zero  $x^*=2/3$ with multiplicity 3. Thus, one unique positive equilibrium $\left(\frac{2}{3},\frac{2}{3}\right)$ exists if $(u,v)= (3,3)$. But, if $(u,v)\neq (3,3)$ on the curve $uv=0$, then $T_1=0$, $T_1'=0$, and $T''_1\neq 0$. In this case, there are three real zeros of $T_1$, among which two are equal. This means that there are two distinct positive equilibria. The proof is complete.
\end{proof}

We summarize the above results in Figure \ref{fig:num-equilibria}. The regions where three and one unique positive equilibria exist are marked in dark-gray and light-gray, respectively.

\begin{figure}[htbp]
    \centering
    \includegraphics[width=0.4\textwidth]{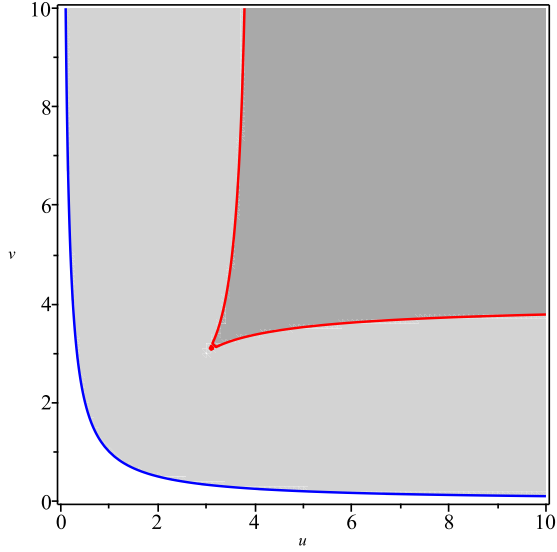}
    \caption{Partitions of the parameter set $\{(u,v)\,|\,u,v>0\}$ of map \eqref{eq:map-kopel} for distinct numbers of equilibria. The blue and red curves are defined by $uv=1$ and $R_1=0$, respectively. In the dark-gray and light-gray regions, there are three and one positive equilibria, respectively.}
    \label{fig:num-equilibria}
\end{figure}

\section{Local Stability}

The Jacobian matrix of map \eqref{eq:map-kopel} at the equilibrium $(x^*,y^*)$ is 

\begin{equation}
J=\left[
\begin{matrix}
1-a &  u a (1-2\,y^*)
\\
 v b (1-2\,x^*) & 1-b 
\end{matrix}
\right].
\end{equation}
Then, the characteristic polynomial of $J$ becomes
$$CP(\lambda)=\lambda^2-\Tr(J)\lambda+\Det(J),$$
where
$$\Tr(J)=2-a-b,~~\Det(J)=(1-a)(1-b)-uvab(1-2\,x^*)(1-2\,y^*),$$ 
are the trace and the determinant of $J$, respectively.
According to the Jury criterion \cite{Jury1976I}, the conditions for the local stability include:
\begin{enumerate}
	\item $CD_1^J\equiv CP(1)= 1-\Tr(J)+\Det(J)>0$,
	\item $CD_2^J\equiv CP(-1)= 1+\Tr(J)+\Det(J)>0$,
	\item $CD_3^J\equiv 1-\Det(J)>0$.
\end{enumerate}

For the zero equilibrium $E_0(0,0)$, we have
\begin{align*}
	&CD_1^J(E_0)=a b (1-u v),\\
	&CD_2^J(E_0)= a b ( 1-u v )-2 a-2 b+4,\\
	&CD_3^J(E_0)= a b (u v-1)+a+b.
\end{align*}
One can easily derive that $E_0$ is locally stable if $1-\frac{a+b}{ab}<uv<1$. In what follows, we focus on positive equilibria. 

It is known that the discrete dynamic system may undergo a local bifurcation when the equilibrium loses its stability at $CD_1^J=0$, $CD_2^J=0$, or $CD_3^J=0$. Denote $\pset{T}=[T_2,T_1]$. Firstly, we consider the possibility of $CD_1^J=0$ by computing
$$\res(CD_1^J,\pset{T})\equiv \res(\res(CD_1^J,T_2,y^*),T_1,x^*).$$
It is obtained that 
$$\res(CD_1^J,\pset{T})=-a^{3} b^{3} u^{3} v^{6} \left(u v-1\right) R_1.$$
According to Lemma \ref{lem:res-com}, $CD_1^J$ and $\pset{T}=0$ imply $\res(CD_1^J,\pset{T})=0$. Hence, at a positive equilibrium of map \eqref{eq:map-kopel}, a necessary condition for $CD_1^J=0$ is $uv=1$ or $R_1=0$. Similarly, regarding the possibility of $CD_2^J=0$ and $CD_3^J=0$, we compute
$$\res(CD_2^J,\pset{T})=-u^3 v^6 R_2,~~\res(CD_3^J,\pset{T})=u^3 v^6 R_3,$$
where
\begin{align*}
	\begin{autobreak}
R_2=
a^{3} b^{3} u^{3} v^{3} 
- 4\,a^{3} b^{3} u^{3} v^{2} 
- 4\,a^{3} b^{3} u^{2} v^{3} 
+ 17\,a^{3} b^{3} u^{2} v^{2} 
+ 4\,a^{3} b^{3} u^{2} v 
+ 4\,a^{3} b^{3} u v^{2} 
- 2\,a^{3} b^{2} u^{2} v^{2} 
- 2\,a^{2} b^{3} u^{2} v^{2} 
- 45\,a^{3} b^{3} u v 
+ 8\,a^{3} b^{2} u^{2} v 
+ 8\,a^{3} b^{2} u v^{2} 
+ 8\,a^{2} b^{3} u^{2} v 
+ 8\,a^{2} b^{3} u v^{2} 
+ 4\,a^{2} b^{2} u^{2} v^{2} 
- 36\,a^{3} b^{2} u v 
- 36\,a^{2} b^{3} u v 
- 16\,a^{2} b^{2} u^{2} v 
- 16\,a^{2} b^{2} u v^{2} 
+ 27\,a^{3} b^{3} 
- 4\,a^{3} b u v 
+ 64\,a^{2} b^{2} u v 
- 4\,a b^{3} u v 
+ 54\,a^{3} b^{2} 
+ 54\,a^{2} b^{3} 
+ 16\,a^{2} b u v 
+ 16\,a b^{2} u v 
+ 36\,a^{3} b 
- 36\,a^{2} b^{2} 
+ 36\,a b^{3} 
- 16\,a b u v 
+ 8\,a^{3} 
- 120\,a^{2} b 
- 120\,a b^{2} 
+ 8\,b^{3} 
- 48\,a^{2} 
+ 48\,a b 
- 48\,b^{2} 
+ 96\,a 
+ 96\,b 
- 64,
\end{autobreak}\\

\begin{autobreak}
R_3=
a^{3} b^{3} u^{3} v^{3} 
- 4\,a^{3} b^{3} u^{3} v^{2} 
- 4\,a^{3} b^{3} u^{2} v^{3} 
+ 17\,a^{3} b^{3} u^{2} v^{2} 
+ 4\,a^{3} b^{3} u^{2} v 
+ 4\,a^{3} b^{3} u v^{2} 
- a^{3} b^{2} u^{2} v^{2} 
- a^{2} b^{3} u^{2} v^{2} 
- 45\,a^{3} b^{3} u v 
+ 4\,a^{3} b^{2} u^{2} v 
+ 4\,a^{3} b^{2} u v^{2} 
+ 4\,a^{2} b^{3} u^{2} v 
+ 4\,a^{2} b^{3} u v^{2} 
- 18\,a^{3} b^{2} u v 
- 18\,a^{2} b^{3} u v 
+ 27\,a^{3} b^{3} 
- a^{3} b u v 
- 2\,a^{2} b^{2} u v 
- a b^{3} u v 
+ 27\,a^{3} b^{2} 
+ 27\,a^{2} b^{3} 
+ 9\,a^{3} b 
+ 18\,a^{2} b^{2} 
+ 9\,a b^{3} 
+ a^{3} 
+ 3\,a^{2} b 
+ 3\,a b^{2} 
+ b^{3}.
\end{autobreak}
\end{align*}

Therefore, local bifurcations of the equilibria may take place when $uv=1$ or $R_i=0$, $i=1,2,3$. However, we should mention that the signs of $CD_i^J$ and $\res(CD_i^J,\pset{T})$ may not be the same for a given $i$. Even if the signs of $\res(CD_i^J,\pset{T})$ are fixed, the signs of $CD_i^J$ may change. For example, if keeping $a=b=1$, we have $\res(R_1,\pset{T})<0$, $\res(R_3,\pset{T})<0$ and $\res(R_3,\pset{T})>0$ when $u=v=4$ or $u=v=\frac{13}{4}$. However, when $u=v=4$, there are three equilibria, i.e., $\left(\frac{3}{4},\frac{3}{4}\right)$, $\left(\frac{5}{8}-\frac{\sqrt{5}}{8},\frac{5}{8}+\frac{\sqrt{5}}{8}\right)$, and $\left(\frac{5}{8}+\frac{\sqrt{5}}{8},\frac{5}{8}-\frac{\sqrt{5}}{8}\right)$, where the signs of $CD_3^J$ are $+$, $-$, and $-$, respectively. Whereas, when $u=v=\frac{13}{4}$, there are three equilibria $\left(\frac{9}{13},\frac{9}{13}\right)$, $\left(\frac{17}{26}-\frac{\sqrt{17}}{26},\frac{17}{26}+\frac{\sqrt{17}}{26}\right)$, and $\left(\frac{17}{26}+\frac{\sqrt{17}}{26},\frac{17}{26}-\frac{\sqrt{17}}{26}\right)$, where the signs of $CD_3^J$ are $+$, $+$, and $+$, respectively.

\begin{remark}\label{rem:idea}
However, it can be derived that if we vary the parameter point continuously without going across the algebraic variety $\res(CD_i^J,\pset{T})=0$ in the parameter space, then $CD_i^J$ will not annihilate (become zero) and its sign will keep fixed. This means that in each region divided by $\res(CD_i^J,\pset{T})=0$, $i=1,2,3$, the signs of $CD_i$, $i=1,2,3$, will not change and we can determine the number of real zeros satisfying $CD_i>0$, $i=1,2,3$, by just selecting some sample point. 	
\end{remark}

Firstly, we consider the case of $a=b=1$. This case is much simpler but is of tremendous economic interest itself because the best response reaction functions are actually adopted by the two duopolists. We denote 
\begin{align*}
	&S_2\equiv R_2|_{a=b=1}=\left(u v-1\right) R_1,\\
	&S_3\equiv R_3|_{a=b=1}=u^{3} v^{3}-4\, u^{3} v^{2}-4\, u^{2} v^{3}+15\, u^{2} v^{2}+12\, u^{2} v+12\, u \,v^{2}-85\, u v+125.
\end{align*}
According to Remark \ref{rem:idea}, we acquire the regions for distinct numbers of stable positive equilibria, which are depicted in Figure \ref{fig:num-stable-eq}. We omit the calculation details due to space limitations. In Figure \ref{fig:num-stable-eq}, the region of two stable positive equilibria is marked in yellow, while that of one stable positive equilibrium is marked in light-gray. Indeed, it may be quite complicated to algebraically describe a region bounded by algebraic varieties. For example, the yellow region can not be described by inequalities of $uv-1$, $R_1$, and $S_3$ (their signs at the yellow region are the same to those at the white region where $(4,4)$ lies). However, this region can be described algebraically by introducing additional polynomials, which can be found in the so-called generalized discriminant list. Readers may refer to \cite{Li2014C, Yang2001A} for more information.

We formally summarize the obtained results in Theorem \ref{thm:stable-equiv}, where $A_1$ and $A_2$ are additional polynomials picked out from the generalized discriminant list. It should be mentioned the computations of searching for these polynomials are quite expensive although this process works perfectly in theory. 

\begin{figure}[htbp]
  \centering
  \subfloat[$(u,v)\in (0,10)\times(1,10)$]{\includegraphics[width=0.4\textwidth]{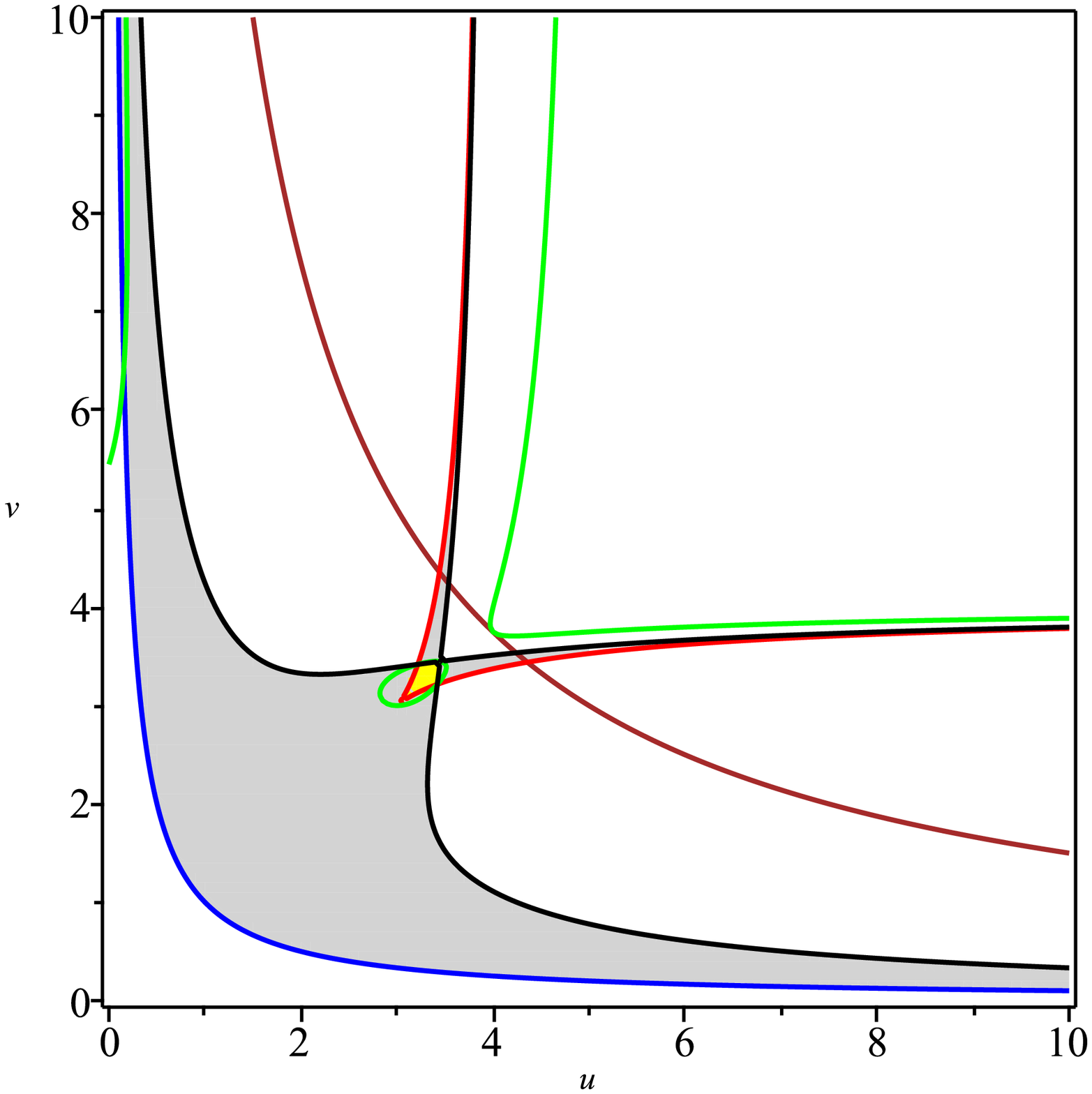}} 
  \subfloat[$(u,v)\in (2.5,5)\times(2.5,5)$]{\includegraphics[width=0.4\textwidth]{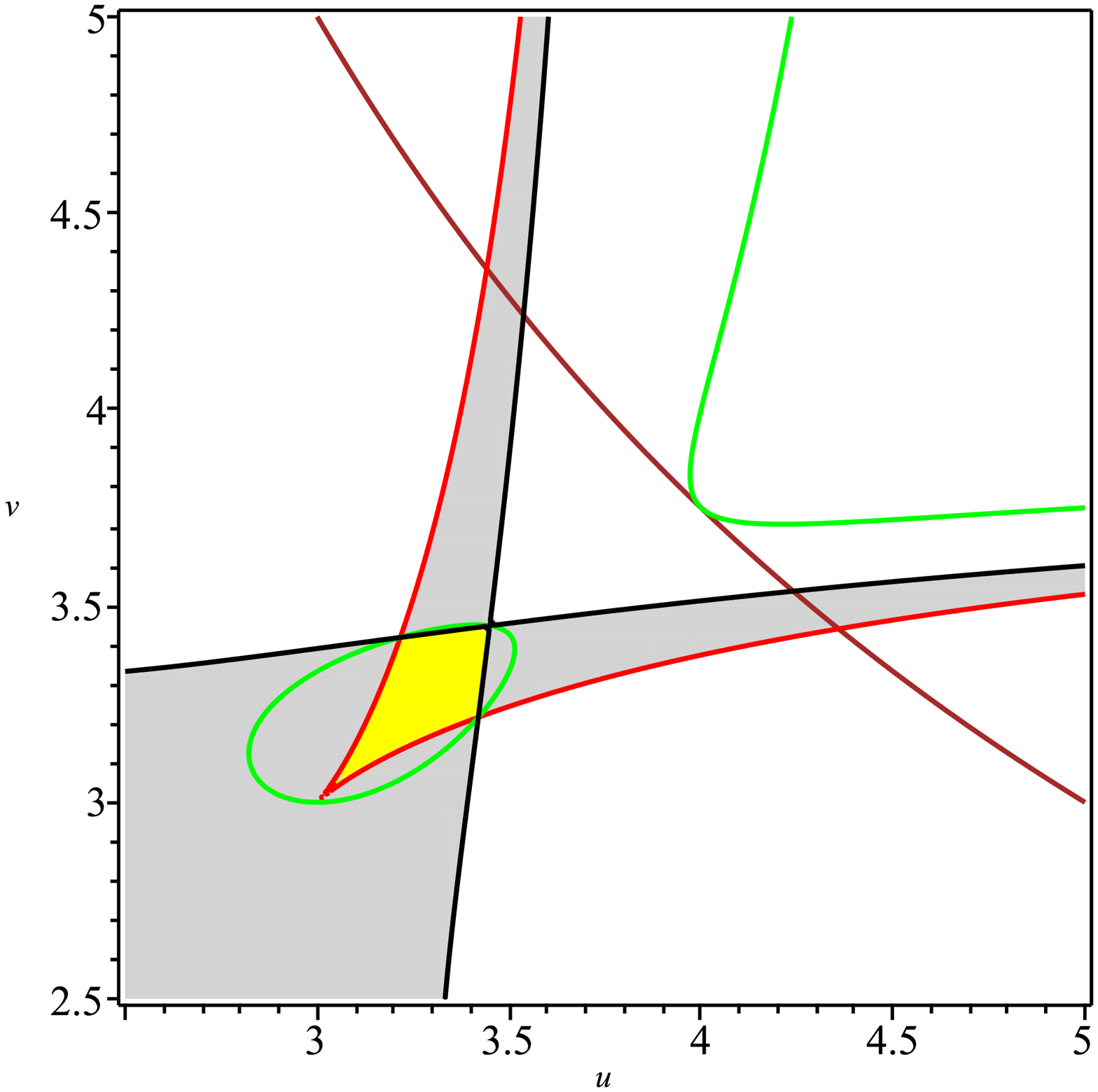}} \\

  \caption{Partitions of the parameter set $\{(u,v)\,|\,u,v>0\}$ in the case of $a=b=1$ for distinct numbers of stable positive equilibria. The blue, red, black, brown, and green curves are defined by $uv=1$, $R_1=0$, $S_3=0$, $A_1=0$, and $A_2=0$, respectively. In the yellow and light-gray regions, there are two and one stable positive equilibria, respectively.}
\label{fig:num-stable-eq}
\end{figure}

%

\begin{theorem}\label{thm:stable-equiv}
	In the case of $a=b=1$, two stable positive equilibria exist if $R_1>0$, $S_3>0$, $A_1<0$, and $A_2>0$, where
\begin{align*}
&A_1=uv-15,~~A_2=u^2v^2-4\,u^2 v-5\, uv^2+21\,uv+11\,v-60.
\end{align*}
Furthermore, one unique stable positive equilibrium exists if $uv>1$, $R_1<0$, $S_3>0$ or $uv>1$, $R_1>0$, $S_3<0$. 
\end{theorem}

In the more general case of $a=b$, homogeneous adaptive expectations are employed by the two duopolists. Similar calculations can be conducted. Readers can understand that the algebraic description of the region where two stable positive equilibria exist can not be computed in a reasonably short time according to our implementations of the methods. However, we can geometrically describe the region of two stable positive equilibria: it is bounded by $R_1=0$ (the red surface), $H_3=0$ (the blue surface), $a=0$, and $a=1$ (See Figure \ref{fig:3d-2stable}). For the region where one stable positive equilibrium exists, our computations can yield the results, which are reported in Theorem \ref{thm:ab-one}.

\begin{figure}[htbp]
    \centering
    \includegraphics[width=10cm]{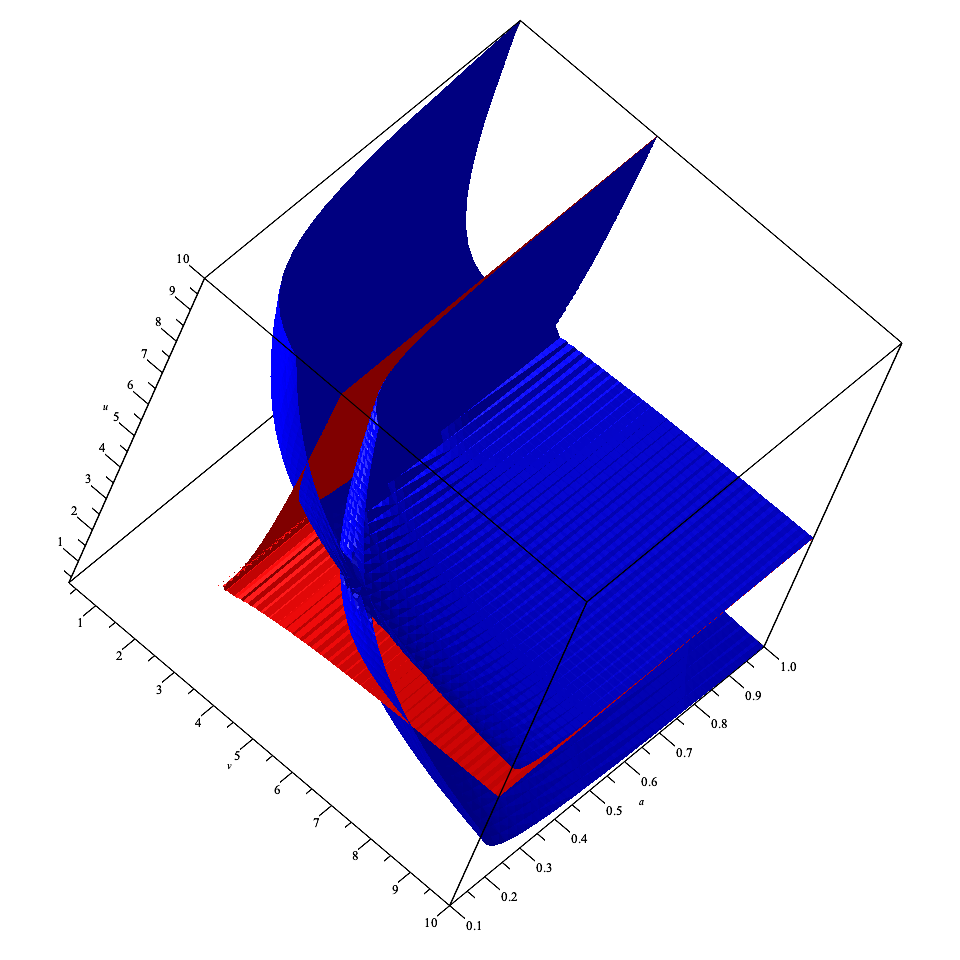}
    \caption{The parametric set $\{(u,v,a)\,|\,u,v>0,0<a\leq 1\}$ of map \eqref{eq:map-kopel} in the case of $a=b$. The surfaces of $R_1=0$ and $H_3=0$ are colored in red and blue, respectively. The region where two stable positive equilibria exist is bounded by the surfaces $R_1=0$, $H_3=0$, $a=0$, and $a=1$.}
    \label{fig:3d-2stable}
\end{figure}

\begin{theorem}\label{thm:ab-one}
In the case of $a=b$, one unique stable positive equilibrium exists if $uv>1$, $R_1<0$, $H_3>0$ or $uv>1$, $R_1>0$, $H_3<0$, where $H_3\equiv R_3|_{b=a}$, i.e.,
\begin{align*}
\begin{autobreak}
H_3=
a^{3} u^{3} v^{3}
-4\,a^{3} u^{3} v^{2}
-4\,a^{3} u^{2} v^{3}
+17\,a^{3} u^{2} v^{2}
+4\,a^{3} u^{2} v
+4\,a^{3} u v^{2}
-2\,a^{2} u^{2} v^{2}
-45\,a^{3} u v
+8\,a^{2} u^{2} v
+8\,a^{2} u v^{2}
-36\,a^{2} u v
+27\,a^{3}
-4\,a u v
+54\,a^{2}
+36\,a
+8.
\end{autobreak}
\end{align*}
	
\end{theorem}

In the general case of $a\neq b$, the computations for searching the algebraic description become particularly complicated. Furthermore, the geometric description can not be plotted either since more than 3 parameters are involved. We leave the exploration of this case for our future study.


\section{Concluding Remarks}

This study filled the gap that there are nearly no analytical investigations on the equilibria and their stability in the asymmetric duopoly model of Kopel. We employed several tools based on symbolic computations such as the triangular decomposition and resultant in our investigation. The results produced by symbolic computations are exact and rigorous, and thus can be used in proving theorems related to polynomial equations and algebraic varieties. We derived the possible positions of the equilibria in Kopel's model. Specifically, all the equilibria should lie in $[0,1]\times [0,1]$ provided that $uv\geq 1$ (see Proposition \ref{prop:all-equi}). We explored the possibility of the existence of multiple positive equilibria and established a necessary and sufficient condition for a given number of equilibria to exist (see Theorem \ref{thm:equi-num}). Furthermore, if the two duopolists adopt the best response reactions ($a=b=1$) or homogeneous adaptive expectations ($a=b$), we established rigorous conditions for distinct numbers of positive equilibria to exist for the first time (see Theorems \ref{thm:stable-equiv} and \ref{thm:ab-one}). In the general case of $a\neq b$, however, we fail to obtain the complete conditions that a given number of stable equilibria exist. We explained the essential difficulty is that the algebraic description of a region bounded by algebraic varieties is expensive to compute.

\section*{Acknowledgements}

The authors are grateful to the anonymous referees for their helpful comments. This work has been supported by Philosophy and Social Science Foundation of Guangdong under Grant GD21CLJ01, Major Research and Cultivation Project of Dongguan City University under Grant 2021YZDYB04Z and 2022YZD05R, Social Development Science and Technology Project of Dongguan under Grant 20211800900692.

\bibliographystyle{abbrv}

%
%


%
%
%
%
%
%
%
%
%
%
%
%
%
%
%
%

\end{CJK}
\end{document}